\documentclass{article}
\makeatletter
\let\proof\@undefined
\let\endproof\@undefined
\makeatother

\usepackage{enumerate}
\usepackage{url}
\usepackage{ulem}
\usepackage{subfigure}
\usepackage{a4wide}
\usepackage{color}

\usepackage{amsfonts}
\usepackage{amsthm}

   \usepackage[pdftex]{graphicx}
\usepackage{amsmath}

\usepackage{epstopdf}
\usepackage{color}

\newcommand{\cC}{{\mathcal C}}

\newtheorem{theorem}{Theorem}

\author{Leo van Iersel, Steven Kelk, Nela Leki{\'c}, Leen Stougie}

\begin{document}
\title{A short note on exponential-time algorithms for hybridization number}
\maketitle

\begin{abstract}
\noindent
In this short note we prove that, given two (not necessarily binary) rooted phylogenetic trees $T_1, T_2$ on the same set of taxa $X$, where $|X|=n$, the hybridization number of $T_1$ and $T_2$ can be computed in time $O^{*}(2^n)$ i.e. $O( 2^{n} \cdot poly(n) )$. The result also means that a Maximum Acyclic Agreement Forest (MAAF) can be computed within the same time bound.
\end{abstract}

\section{Introduction}
Let~$X$ be a finite set. A \emph{rooted phylogenetic} $X$-\emph{tree}, henceforth abbreviated to \emph{tree},  is a rooted tree with no vertices with indegree~1 and outdegree~1, a root with indegree~0 and outdegree at least~2, and leaves bijectively labelled by the elements of~$X$.
A \emph{rooted phylogenetic network}, henceforth abbreviated to \emph{network}, is a directed acyclic graph with no vertices with indegree~1 and outdegree~1 and leaves bijectively labelled by the elements of~$X$.

A tree~$T$ is \emph{displayed} by a network~$N$ if~$T$ can be obtained from a subgraph of~$N$ by contracting edges. Note that, when $T$ is not binary, this means that the image of $T$ inside $N$ can be more ``resolved'' than $T$ itself.  Using~$d^-(v)$ to denote the indegree of a vertex~$v$, a \emph{reticulation} is a vertex~$v$ with~$d^-(v)\geq 2$. The \emph{reticulation number} of a network~$N$ with vertex set~$V$ is given by

\[
r(N)=\sum_{v\in V : d^-(v)\geq 2}(d^-(v)-1).
\]
Given two (not necessarily binary) trees $T_1$, $T_2$, the \emph{hybridization number} problem (originally introduced in \cite{baroni05}) asks us to minimize $r(N)$ ranging over all networks that display $T_1$ and $T_2$. 

There has been extensive work on fixed-parameter tractable (FPT) algorithms for the hybridization number problem.
The fastest such algorithm currently works only on binary trees and has a running time of $O( 3.18^{r} \cdot poly(n) )$ where $r$ is the hybridization number and $n=|X|$ \cite{whidden2013fixed}. Given that $n$ is a trivial upper bound on the hybridization number of two trees this immediately yields an exponential-time algorithm with running time $O^{*}( 3.18^{n} )$
for the binary case. In \cite{firststeps} a $O^{*}( 3^n )$ algorithm was presented (again restricted to the binary case). In \cite{elusiveness} a $O^{*}( 2^n )$ algorithm was implied but this relied on the claimed equivalence between the softwired cluster model and the model described in \cite{bafnabansal2006}, which was not formally proven. Here we describe explicitly a $O^{*}(2^n)$ algorithm that does not rely on this equivalence. This also means that a
Maximum Acyclic Agreement Forest (MAAF) can be computed within the same time bound (see 
e.g. \cite{nonbinCK} for related discussions).

For further background and definitions on hybridization number and phylogenetic networks we refer the reader to recent articles such as \cite{terminusest}. For background and definitions on softwired clusters (which the proof below uses heavily) see \cite{elusiveness}.
\pagebreak
\section{Results}

\begin{theorem}
Let $T_1$ and $T_2$ be two (not necessarily binary) rooted phylogenetic trees on the same set of taxa $X$, where
$|X|=n$. Then the hybridization number $h(T_1, T_2)$ can be computed in time $O^{*}(2^{n})$.
\end{theorem}
\begin{proof}
Let $\cC = Cl(T_1) \cup Cl(T_2)$ be the union of the sets of clusters induced by the edges of the trees $T_1$ and $T_2$. It has been shown
that $r(\cC)$, the minimum reticulation number of a phylogenetic network representing all the clusters in $\cC$, is exactly equal to $h(T_1, T_2)$ \cite[Lemma 12]{elusiveness} and that optimal solutions for one problem can be transformed in polynomial time into
optimal solutions for the other \cite{terminusest}. We hence focus on computation of $r(\cC)$. Recall that an ST-set $S$ of a set of clusters is a subset of $X$ such that $S$ is compatible with every cluster in $\cC$, and such that all clusters in $\cC|S$
are pairwise compatible, where $\cC|S = \{ C \cap S : C \in \cC \}$. (The non-empty ST-sets are in one-to-one correspondence with common pendant subtrees of $T_1$ and $T_2$ \cite{terminusest}). For $X' \subseteq X$, we write $\cC \setminus X'$ to denote $\{ C \setminus X' : C \in \cC \}$.
An ST-set sequence of length $k$ is a sequence $S_1, S_2, \ldots, S_k$ such that each $S_i$ is an ST-set of $\cC_{i-1}$,
where $\cC_{0} = \cC$ and for $1 \leq i \leq k$,  $\cC_i = \cC_{i-1} \setminus S_{i}$. Such a sequence is a \emph{tree} sequence if $\cC_k$ is compatible. Note that if $\cC$ is compatible
then this is characterized by the empty tree sequence and we say that $k=0$. The value $r(\cC)$ is equivalent to the minimum possible length ranging over all ST-set tree sequences \cite[Corollary 9]{elusiveness}. Without loss of generality we can assume that $S_i$ is a \emph{maximal} ST-set sequence i.e. where each $S_i$ is a maximal ST-set of $\cC_{i-1}$. For a given set of clusters on $n$ taxa there are at most $n$ maximal ST-sets, they partition the set of taxa and they can be computed in polynomial time \cite{elusiveness}.
Clearly, $r(\cC) = 0$ if $\cC$ is compatible which can be checked in polynomial time. Otherwise the above observations yield the
following expression, where $ST(\cC)$ is the set of maximal ST-sets of $\cC$:
\begin{equation}
r(\cC) = \min_{S \in ST(\cC)} \bigg ( 1 + r(\cC \setminus S) \bigg )
\end{equation}
This can be computed in time $O^{*}(2^n)$ by standard exponential time dynamic programming. That is, compute $r(\cC)$ by computing $r(\cC | X')$ for all possible $\emptyset \subset X' \subset X$, increasing the cardinality of $X'$ from small to large. Each $r(\cC | X')$ can then be computed by consulting at most $n$ smaller subproblems. This yields an overall running time of $O(2^n \cdot poly(n))$.
\end{proof}

\section{Discussion}

A consequence of the above analysis is that, when solving hybridization number, there are at most $2^{n}$ relevant subproblems and each such subproblem can be characterized by a subset of $X$. Any algorithm that attempts to compute the hybridization number by iteratively pruning maximal common pendant subtrees (equivalently, maximal ST-sets) until the input trees are compatible, can thus easily attain a $O^{*}(2^n)$ upper bound on its running time, at the expense of potentially consuming exponential space. That is, by storing the solutions to subproblems in a look-up table (i.e. hashtable), indexed by the subset of $X$ that characterises the subproblem. 

Finally, an obvious open question that remains is whether the hybridization number of two trees can be computed in time $O^{*}( c^n )$ for any constant $c < 2$.

\bibliographystyle{plain}
\bibliography{MAAFexptime}

\end{document}